\documentclass{article}

\usepackage{hyperref}
\usepackage{amsmath,amssymb,amsfonts,amsthm}
\usepackage{xspace}
\usepackage{graphicx}
\usepackage{color}
\usepackage{latexsym}
\usepackage{paralist}
\usepackage{algorithm}
\usepackage{algorithmicx}
\usepackage{algpseudocode}
\usepackage{ifpdf}
\usepackage{microtype}
\usepackage{lipsum}
\usepackage{amsfonts}
\usepackage{graphicx}
\usepackage{epstopdf}
\usepackage{thmtools, thm-restate}
\usepackage{bm}
\usepackage[T1]{fontenc}
\usepackage{paralist}
\usepackage{lmodern}
\newtheorem{theorem}{Theorem}

\newtheorem{lemma}[theorem]{Lemma}

\newcommand{\prob}[1]{\mathbf{Pr}[#1]}

\title{Leader Election in Multi-Hop Radio Networks\footnote{Research partially supported by the Centre for Discrete Mathematics and its Applications (DIMAP).}}

\begin{document}
	

	
	\author{\textbf{Artur Czumaj} \hspace{4mm} \textbf{Peter Davies} \\[0.10in]
		Department of Computer Science \\
		Centre for Discrete Mathematics and its Applications 
		\\
		University of Warwick}
	
	\title{\textbf{Leader Election in Multi-Hop Radio Networks}
		\thanks{Research partially supported by the Centre for Discrete Mathematics and its Applications (DIMAP).}
		\thanks{Contact information: \{A.Czumaj, P.Davies.4\}@warwick.ac.uk. Phone: +44 24 7657 3796.}
	}
	
	\maketitle

\begin{abstract}
In this paper we present a framework for leader election in multi-hop radio networks which yield randomized leader election algorithms taking $O(\text{broadcasting time})$ in expectation, and another which yields algorithms taking fixed $O(\sqrt{\log n})$-times broadcasting time. Both succeed with high probability.

We show how to implement these frameworks in radio networks without collision detection, and in networks with collision detection (in fact in the strictly weaker beep model). In doing so, we obtain the first optimal expected-time leader election algorithms in both settings, and also improve the worst-case running time in directed networks without collision detection by an $O(\sqrt {\log n})$ factor.
\end{abstract}

\section{Introduction}
\emph{Leader election} is the problem of ensuring that all nodes agree on a single node to be designated leader. Specifically, at the conclusion of a leader election algorithm, all nodes should output the same ID, and precisely one node should identify this ID as its own. Since we are concerned with randomized algorithms, we can assume that nodes generate their own $\Omega(\log n)$-bit IDs uniformly at random at the beginning of the algorithm, and achieve uniqueness with high probability\footnote{Throughout this paper, we say that an event $\mathcal{E}$ holds \emph{with high probability (w.h.p.)} if $\prob{\mathcal{E}} \ge 1-n^{-c}$ for any constant $c>0$.}. Leader election is a fundamental primitive in distributed computations and, as the most fundamental means of breaking symmetry within radio networks, it is used as a preliminary step in many more complex communication tasks. For example, many fast multi-message communication protocols require construction of a breadth-first search tree (or some similar variant), which in turn requires a single node to act as source (for more examples, cf. \cite{-CKP12,-FSW14,-GH13}, and the references therein).

In this paper we present a framework which yields randomized leader election algorithms taking fixed $O(\sqrt{\log n})$-times broadcasting time, and another which yields algorithms taking $O(\text{broadcasting time})$ in expectation. Both succeed with high probability. We then implement these frameworks in radio networks with and without collision detection.

\subsection{Communication Models}
We consider three related distributed communication models (cf. e.g. \cite{-GH13,-BGI91,-CK10}): \emph{radio networks without collision detection}, \emph{radio networks with collision detection}, and the \emph{beep model}. All of these models have the same underlying network representation: a communications network is modeled by a (directed or undirected) graph $\mathcal{N} = (V,E)$, where the set of nodes corresponds to the set of transmitter-receiver stations. A directed edge $(v,u) \in E$ means that the node $v$ can send a message to the node $u$, whereas an undirected edge $\{v,u\} \in E$ means that the nodes $v$ and $u$ can exchange messages in both directions. To make the leader election problem feasible, we assume that every pair of nodes in $\mathcal{N}$ is strongly connected.

We consider ad-hoc network models, which means that we assume that a node does not have any  prior knowledge about the topology of the network, its in-degree and out-degree, or the set of its neighbors. We assume that the only knowledge of each node is the size of the network $n$, and the \emph{eccentricity} of the network $D$, which is the maximum distance between any pair of nodes in $\mathcal{N}$. The assumption of knowledge of $n$ and $D$ is integral to the results in this paper, and is common in the study of multi-hop radio networks (see, e.g., \cite{-Pel07} and the references therein).

We assume that all nodes have access to a global clock and work \emph{synchronously} in discrete time steps. When we refer to the ``running time'' of an  algorithm, we mean the number of time steps which elapse before completion (i.e., we are not concerned with the number of calculations nodes perform within time steps).

All of our algorithms will be designed to succeed with high probability. We will study two different measures of running time behavior: some of our algorithms will have running time which is a random variable, and we will analyze the expectation of this random variable. We will call these variable-time algorithms. Others will have fixed running time bound, which they adhere to even in the worst case. This latter case is sometimes referred to in the literature (e.g. in \cite{-GH13}) as high-probability running time; since we are interested in success with high probability, these concepts are equivalent. This is because an algorithm with a variable running time which is less than $t$ with high probability can be curtailed after $t$ steps while still ensuring success with high probability.

The difference between the three models is in the communication rules for nodes, as we will discuss next in Sections \ref{subsubsec:radio-networks-without-collision}--\ref{subsubsec:beeps}.

\begin{table}[t]\small
	\begin{minipage}{\textwidth}
	\renewcommand{\arraystretch}{1.4}
	\centerline{
		\begin{tabular}{|c|l|}
			\hline
			&
			\multicolumn{1}{|c|}{Radio networks without collision detection (directed networks)}
			\\\hline
			Expected complexity &
			$\Omega(D \log\frac{n}{D}+\log^2n)$ \cite{-GH13,-ABLP91,-KM98}
			\\
			&
			$O((D \log\frac{n}{D}+\log^2n) \log\log n)$ \cite{-BGI91,-CR06J,-KP03b,-W86}
			\\
			&
			$O(D \log\frac{n}{D}+\log^2n)$ \textbf{[Theorem \ref{th:var}]}
			\\
			\hline
			Worst-case &
			$\Omega(D \log\frac{n}{D}+\log^2n)$ \cite{-GH13,-ABLP91,-KM98}
			\\
			complexity &
			$O((D \log\frac{n}{D}+\log^2n) \log n)$ \cite{-BGI91,-CR06J,-KP03b}
			\\
			&
			$O((D \log\frac{n}{D}+\log^3n) \min\{\log\log n, \log\frac{n}{D}\})$\footnote{Undirected networks only}\cite{-GH13}
			\\
			&
			$O((D \log\frac{n}{D}+\log^2n)\sqrt{\log n})$ \textbf{[Theorem \ref{th:fix}]}\\
			\hline
		\end{tabular}
	} 
	\caption{\small Summary of algorithms in radio networks \emph{\textbf{without} collision detection}.}
	\label{table:NOCD}
\end{minipage}
\end{table}

\subsubsection{Radio Networks Without Collision Detection}
\label{subsubsec:radio-networks-without-collision}
The classical model of \emph{ad-hoc multi-hop radio networks without collision detection} is the most well-studied of the models (cf. \cite{-BGI91}). Its defining feature is that in each time step a node can either \emph{transmit} an $O(\log n)$-bit message to all of its out-neighbors at once or can remain silent and \emph{listen} to the messages from its in-neighbors. If a node $v$ listens in a given round and precisely one of its in-neighbors transmits, then $v$ receives the message. In all other cases $v$ receives nothing; in particular, the lack of collision detection means that $v$ is unable to distinguish between zero of its in-neighbors transmitting and more than one.

\subsubsection{Radio Networks With Collision Detection}
\label{subsubsec:radio-networks-with-collision}
The model \emph{with collision detection} (cf. \cite{-GH13}) differs in that a listening node can distinguish between hearing zero transmissions and hearing more than one. However, in the latter case it still does not receive any information about the transmitted messages.

\begin{table}[t]\small
	\begin{minipage}{\textwidth}
	\renewcommand{\arraystretch}{1.4}
	\centerline{
		\begin{tabular}{|l|l|}
			\hline
			&
			\multicolumn{1}{|c|}{Radio networks with collision detection / beep model}
			\\
			&(undirected networks only)\\
			\hline
			Expected complexity &
			$\Omega(D + \log n)$ \cite{-ABLP91,GLS12,-NO02}
			\\
			&
			$O((D \log\frac{n}{D}+\log^2n) \log\log n)$ \cite{-BGI91,-CR06J,-KP03b,-W86}
			\\
			&
			$O((D + \log n\log\log n) \min\{\log\log n, \log\frac{n}{D}\})$\footnote{Achieves same worst-case complexity} \cite{-GH13}
			\\
			&
			$O(D + \log n)$ \textbf{[Theorem \ref{th:cd}]}
			\\
			\hline
		\end{tabular}
	} 
	\caption{\small Summary of algorithms in radio networks \emph{\textbf{with} collision detection}}
	\label{table:CD}
\end{minipage}
\end{table}

\subsubsection{Beep Model}
\label{subsubsec:beeps}
The \emph{beep model} (cf. \cite{-CK10}) is strictly weaker than the model of radio networks with collision detection. Instead of $O(\log n)$-bit messages, transmitting nodes emit only a `beep'. A listening node hears a beep if at least one of its neighbors beeped, and hears silence otherwise. In particular, it cannot distinguish between one neighbor beeping and multiple neighbors beeping.

\subsection{Previous Work}
The study of leader election in radio networks started in the 1970s with the \emph{single-hop network model}, in which all nodes are directly reachable by all others (in a single hop).

While the complexity of leader election in single-hop networks is now well understood (cf. \cite{-KM98,-W86,GLS12,-NO02,-CMS03,-GW85,-KP13}), the complexity of the problem in the more general model of \emph{multi-hop networks} has been less developed.

In the seminal work initiating the study of the complexity of communication protocols in multi-hop radio networks, Bar-Yehuda et al.\ \cite{-BGI91} developed a general randomized framework of simulating single-hop networks with collision detection by multi-hop networks without collision detection. The framework yields leader election algorithms for multi-hop networks (in directed and undirected networks) running in $O(T_{BC} \cdot \log\log n)$ expected time and $O(T_{BC} \cdot \log n)$ time w.h.p., where $T_{BC}$ is the time required to broadcast a message from a single source to the entire network. The same authors also gave a randomized broadcasting algorithm running in $O(D \log n+\log^2 n)$ time w.h.p., thereby yielding a leader election algorithm taking $O((D \log  n + \log^2 n) \log\log n)$ expected time and $O(D \log^2 n + \log^3 n)$ time w.h.p.

The next improvement came with faster algorithms for broadcast due to Czumaj and Rytter \cite{-CR06J}, and independently Kowalski and Pelc \cite{-KP03b}, which require only $O(D \log \frac{n}{D} + \log^2 n)$ time w.h.p. (see also \cite{-CGGPR00,-CGR00}). Combining these algorithms with the simulation framework of Bar-Yehuda et al., one obtains leader election algorithms (even in the model without collision detection, and in both directed and undirected graphs) running in $O((D \log \frac{n}{D} + \log^2 n)\log \log n)$ expected time and $O((D \log \frac{n}{D} + \log^2 n)\log n)$ time w.h.p.

Very recently, Ghaffari and Haeupler \cite{-GH13} took a new approach, which yielded faster leader election algorithms in \emph{undirected networks}. The main idea of this work is to randomly select a small (logarithmic) number of candidates for the leader and then repeatedly run ``debates'' to reduce the number of candidates to one. Standard random sampling technique allows one to choose in constant time a random set of $\Theta(\log n)$ candidates, with high probability. Then, by running a constant number of broadcasting computations and neighborhood exploration algorithms (this phase is called a ``debate'' in \cite{-GH13} and it relies heavily on the assumption that the network is undirected), one can reduce the number of candidates by a constant factor. Using this approach, Ghaffari and Haeupler gave a leader election algorithm (in undirected networks) that in $O((D \log \frac{n}{D} + \log^3 n) \cdot \min\{\log \log n, \log \frac{n}{D}\})$ rounds elects a single leader w.h.p. For the model with collision detection (and, in fact, the weaker \emph{beep model} also employed in this paper), the same work \cite{-GH13} used a similar approach to elect a leader in $O((D + \log n \log\log n) \cdot \min\{\log \log n, \log \frac{n}{D}\})$ rounds w.h.p. These algorithms are nearly optimal, and are the fastest currently known for undirected networks and for worst-case running time.

\section{Leader election algorithms}

We first give frameworks for leader election in radio network-like scenarios which are independent of the specific communication rules of the models involved.

\subsection{\textsc{Verify} and \textsc{Compete}}
Our leader election frameworks will be based upon the use of two sub-procedures, which we will call $\textsc{Verify}$ and $\textsc{Compete}$. In broad terms, these procedures are both means of utilizing a global broadcast to collect some information about the current state of a leader election attempt.

\paragraph{\textsc{Verify}}
The purpose of \textsc{Verify} is to determine whether the input set (which in our application will be a set of candidate leaders) is of size $0$, $1$, or greater than $1$. Further, in the case that the set is of size exactly $1$ (which is what we hope for), all nodes receive the ID of its sole member.

Formally, $\textsc{Verify}(C)$ takes as input a subset of nodes $C$ and outputs a pair $(m(v),b)\in [L] \times \{0,1,2\}$ to each node $v$, satisfying the following conditions:

\begin{itemize}
\item $b = \begin{cases}0 & \text{ if $|C| = 0$;}\\ 1 & \text{ if $|C| = 1$;} \\ 2 & \text{ otherwise;}\end{cases}$
\item if $C$ consists of only a single element $c$, then $m(v) = ID(c)$ for all nodes~$v$.
\end{itemize}

Here $b$ carries the information about the size of $C$, and if $|C| = 1$ as required, $m(v)$ gives the candidate ID. Otherwise, we make no restrictions on $m(v)$.

\paragraph{\textsc{Compete}}
$\textsc{Compete}$ has stronger requirements than $\textsc{Verify}$. Rather than just returning information about the size of the input set, it instead returns an output set which is of size at least $1$, and strictly smaller than the input set (as long as this was of size greater than $1$). In our application, this will allow us to thin out a set of candidate leaders by (at least) one per iteration, until we reach a single leader.

Formally, $\textsc{Compete}(C)$ takes as input a subset of nodes $C$ and outputs a pair $(m(v),C') \in [L]\times 2^V$ to each node $v$, satisfying the following conditions:
\begin{itemize}
\item $1 \le |C'|<\max \{|C|,2\}$;
\item if $C$ consists of only a single element $c$, then $m(v) = ID(c)$ for all nodes~$v$.
\end{itemize}

In both of these procedures, sets given as input or output are implicit; that is, that each node receives, or outputs, only the information of whether it is itself a member of the set (rather than full knowledge of the set). Furthermore, we will be considering randomized implementations of these procedures, and so will ensure that our implementations meet the specified requirements with high probability.

We will show how to implement these sub-procedures in our network models later; first we describe how to use them to build general leader election frameworks.

\subsection{Leader election frameworks}
In this section we will show how the \textsc{Verify} and \textsc{Compete} procedures can be combined into frameworks for leader election. While in this work we are focused on radio networks, in general these frameworks could extend to other distributed computing models.

\subsubsection{Variable time}
We first give the framework for a leader election algorithm whose running time is a random variable, which is $O($broadcasting time$)$ in expectation.

The idea is simple: we repeatedly randomly choose a set of candidate leaders, and terminate when the set we chose is of size $1$.

\begin{algorithm}[h]
\caption{Leader Election, variable time}
\label{alg:LE}
\small
\begin{algorithmic}
\Loop
	\State each $v\in V$ chooses to be in the set $C$ of \textbf{candidates} with probability $\Theta(\frac1n)$
	\State $(m(V),b) \gets \textsc{Verify}(C)$
    \State \textbf{if} $b=1$ \textbf{then} output $m(V)$, terminate
\EndLoop
\end{algorithmic}
\end{algorithm}

\begin{theorem}\label{th:vartm}
If $\textsc{Verify}(C)$ is implemented in $t$ time to succeed with high probability, then Algorithm \ref{alg:LE} performs leader election in $O(t)$ expected time with high probability.
\end{theorem}

\begin{proof}
Assume for the sake of the analysis that if the algorithm has run unsuccessfully for $\frac n2$ iterations, it terminates. With probability at least $1-\frac{n^{-1}}{2}$, $\textsc{Verify}$ returns the correct result for these $\frac n2$ iterations. Notice that in any particular iteration, the probability that $|C| = 1$ is bounded above by a constant, and denote this constant $c$. With probability at least $1-c^{\frac n2}$ one of the iterations will have had $|C| = 1$, so the algorithm will correctly perform leader election with probability at least $(1-\frac{n^{-1}}{2})(1-c^{\frac n2}) \ge 1-n^{-1}$. Since $c$ is constant, expected number of iterations until $|C| = 1$ is also a constant $k$. Expected number of iterations until termination is therefore at most $k(1-\frac{n^{-1}}{2}) + \frac n2\frac{n^{-1}}{2} < k+1$, i.e., a constant, and since the running time of each iteration is dominated by that of \textsc{Verify}, expected running time is $O(t)$.
\end{proof}

\subsubsection{Fixed time}
Next we give another framework for leader election whose running time is fixed, and therefore has better worst-case performance. This framework is slightly more complex, and consists of two main phases. In the first, each node chooses to be a candidate with probability $\Theta(\frac{ \log n}{n})$, which ensures that $\Theta(\log n)$ candidates are chosen with high probability. Then we repeatedly have candidates drop out with probability $\frac 12$, and use the \textsc{Verify} procedure to check that we didn't remove all candidates. Doing this $\Theta(x)$ times (where $x$ is some parameter to be fixed presently) ensures that only $\Theta(\frac{\log n}{x})$ candidates remain, w.h.p. Then, in the second phase, we repeatedly use \textsc{Compete} to remove candidates one at a time, until only a single one remains.

\begin{algorithm}[h]
\caption{Leader Election, fixed time}
\label{alg:LE2}
\small
\begin{algorithmic}
\State each $v\in V$ chooses to be in the set $C$ of \textbf{candidates} with probability $\Theta(\frac{ \log n}{n})$
\Loop { $\Theta(x)$ times}
	\State each $v \in C$ chooses to be in $C'$ with probability $\frac 12$
	\State $(m(v),b) \gets \textsc{Verify}(C')$
	\State \textbf{if} $b\neq 0$ then $C \gets C'$
\EndLoop
\Loop { $\Theta(\frac{\log n}{x})$ times}
	\State $(m(v),C) \gets \textsc{Compete}(C)$
\EndLoop
\State output $m(v)$, terminate
\end{algorithmic}
\end{algorithm}

The running time of the algorithm will be dominated by the $\Theta(x)$ calls to \textsc{Verify} and the $\Theta(\frac{\log n}{x})$ calls to \textsc{Compete}. Therefore, if we let $t$ and $u$ be the running times of \textsc{Verify} and \textsc{Compete} respectively, we see that to optimize our overall running time we should set $x = \sqrt{\frac ut \log n}$.

\begin{theorem}\label{th:fixtm}
Algorithm \ref{alg:LE2} performs leader election in $O(\sqrt{tu\log n})$ time and succeeds with high probability.
\end{theorem}

\begin{proof}
With high probability, $\Theta( \log n)$ nodes choose to be candidates in $C$. First we will analyze how many candidates remain after the first loop. Let $y$ be the number of rounds of the loop during which $|C|>\frac{2\log n}{x}$, and consider only these rounds. We call such a round $i$ \emph{successful} if $\frac{5|C|}{6}> |C'|\ge 1$. The probability that any round is not successful is at most:
\begin{align*}
    \prob{|C'|=0} &+ \prob{|C'|\ge\tfrac{5|C|}{6}}
        \le
    \left(\frac 12\right)^{|C|} +
        \binom{|C|}{\frac{5|C|}{6}} \left(\frac 12\right)^\frac{5|C|}{6}
        \\
        &=
    \left(\frac 12\right)^{|C|} +
        \binom{|C|}{\frac{|C|}{6}} \left(\frac {1}{32}\right)^\frac{|C|}{6}
        \le
    \left(\frac 12\right)^{|C|} +
        \left(\frac{e \cdot |C|}{\frac{|C|}{6}}\right)^\frac{|C|}{6}
            \cdot \left(\frac{1}{32}\right)^\frac{|C|}{6}
        \\
        &=
    \left(\frac 12\right)^{|C|} + \left(\frac{3e}{16}\right)^\frac{|C|}{6}
        \le
    2 \times 0.9^{|C|}
        \le
    0.9^{\frac{\log n}{x}}
        \enspace.
\end{align*}

Note that this is still true conditioned on the randomness of all previous rounds, and so the total number of unsuccessful rounds is majorized by a binomially distributed variable $\mathbb{B}\text{in}(y,p)$, where $p=0.9^{\frac{\log n}{x}}$. If $y\ge 20x$ then
\begin{align*}
    \prob{\text{at least $\tfrac{y}{2}$ rounds are unsuccessful}}
        &\le
    \prob{\mathbb{B}\text{in}(y,p) \ge \tfrac y2}
        \le
    e^{-\frac 12 y ( \ln \frac {1}{2p} + \ln \frac {1}{2(1-p)})}
        \\
        &\le
    e^{-\frac y2 \ln \frac {0.9^{-\frac{\log n}{x}}}{2}}
        \le
    e^{-\frac y2 \ln e^{\frac{0.1\log n}{x}}}
        \\
        &=
    e^{-\frac {y\log n}{20x}}
        \le
    n^{-\log e} \enspace.
\end{align*}

So either $y< 20x$ (i.e., after $20x$ rounds $|C|<\frac{2\log n}{x}$), or with high probability at least $10x$ of the first $20x$ rounds are successful and so $|C|$ has reduced by a factor of at least $\left(\frac 56\right)^{10x}$, and hence is below $\frac{\log n}{x}$.

Then, since \textsc{Compete} reduces $|C|$ by at least $1$ per round, after $\Theta(\frac{\log n}{x})$ rounds only one candidate remains, and leader election is complete.

The running time is dominated by $\Theta(x)$ rounds of \textsc{Verify} and $\Theta(\frac{\log n}{x})$ rounds of \textsc{Compete}. With $x= \sqrt{\frac{u\log n}{t}}$, this gives a total running time of $O(\sqrt{tu\log n})$.
\end{proof}


\section{Implementation}

We now show how to implement $\textsc{Verify}(C)$ and $\textsc{Compete}(C)$ in the models we consider: radio networks without collision detection, and with collision detection. In the latter case, we will actually use the strictly weaker beep model, as any algorithm that works in the beep model also works in radio networks with collision detection.

\subsection{Radio networks without collision detection}
Our algorithms will make use of some existing techniques for the radio network model, namely methods for broadcasting messages with local neighborhoods and throughout the entire network.

\subsubsection{Local and global broadcast}
For local broadcasting, we will utilize the classical Decay protocol.
The Decay protocol, first introduced by Bar-Yehuda et al. \cite{-BGI92}, is a fundamental primitive employed by many randomized radio network communication algorithms. Its aim is to ensure that, if a node has one or more in-neighbours which wish to transmit a message, it will hear at least one of them. To accomplish this, nodes who do wish to transmit do so with exponentially decaying probability.

\begin{algorithm}[h]
\caption{\textsc{Decay}$(S)$ at a node $v \in S$}
\label{alg:D}
\small
\begin{algorithmic}
\State \textbf{for} $i = 1$ \textbf{to} $\log n$, in time step $i$ \textbf{do}: $v$ transmits its message with probability $2^{-i}$
\end{algorithmic}
\end{algorithm}

The following lemma (cf. \cite{-BGI92}) describes a basic property of \textsc{Decay}, as used in our analysis.

\begin{lemma}
\label{corollary:Decay}
After four rounds of \textsc{Decay}$(S)$, a node $v$ with at least one in-neighbor in $S$ receives a message with probability greater than $\frac12$.
\qed
\end{lemma}

While Decay has a very localized effect, to achieve global tasks we will need more complex primitives. In particular, we will also need a method of propagating information globally. For this we employ another previous result from the literature, which yields a $\textsc{Partial Multi-Broadcast}(S,f)$ algorithm with the following properties:

\begin{itemize}
\item $S \subseteq V$ is a set of \emph{source nodes};
\item $f:S \rightarrow \{0,1\}^{\ell}$, where $\ell = O(\log n)$ is some message length parameter, is a function giving each source a message to broadcast; in our applications, this will either be a node ID, or a single bit ``1'';
\item Each node $v$ interprets some bit-string $m(v)$ upon completion;
    \begin{itemize}
    \item if $S = \emptyset$, then $m(v) = \epsilon$ (the empty string) for all $v$;
    \item if $S \neq \emptyset$, then $\forall v \in V$ $\exists s \in S$ with $m(v) = f(s)$, i.e., each node interprets some source's message.
    \end{itemize}
\end{itemize}

The broadcasting algorithm of Czumaj and Rytter \cite{-CR06J}, performed with every node in $S$ operating as a single source and with nodes interpreting the first transmission they receive to be their output message $m(v)$, yields:

\begin{lemma}[\cite{-CR06J}]
\label{lemma:Partial Multi-Broadcast-algorithm}
There is a Partial Multi-Broadcast algorithm running in time $O(D \log \frac{n}{D} + \log^2 n)$, which succeeds with high probability.
\qed
\end{lemma}

Armed with procedures for both local and global message dissemination, we can implement \textsc{Verify} and \textsc{Compete}:

\subsubsection{Subprocedure Implementation}
In radio networks without collision detection, our implementations of \textsc{Verify} and \textsc{Compete} both run in time $T_{BC} = O(D \log \frac nD+\log^2 n)$ time, and we can in fact combine them into the same process (Algorithm \ref{alg:VCWOCD}; here the output triple contains the necessary outputs for both \textsc{Verify} and \textsc{Compete}). The idea of both is to make use of \textsc{Partial Multi-Broadcast} to inform all nodes of at least one candidate ID, check if any neighboring nodes received different IDs, and performing \textsc{Partial Multi-Broadcast} again to inform the network if this was the case. To meet the stricter requirements of \textsc{Compete}, we have nodes who detected differing IDs broadcast the highest ID they heard in this final phase. Then, if there are multiple candidates, at least one of the candidates (the one with the lowest ID) becomes aware of another with a higher ID and drops out.

We note that this implementation is similar to the method used in the seminal paper by Bar-Yehuda, Goldreich and Itai \cite{-BGI91} to simulate algorithms for single-hop networks within multi-hop networks.

\begin{algorithm}[h]
\caption{\textsc{Verify\&Compete}(C), networks without collision detection}
\label{alg:VCWOCD}
\small
\begin{algorithmic}
\State $m(v) \gets \textsc{Partial Multi-Broadcast}(C,\text{ID}(C))$
\If {$m(v) \neq \epsilon$}
	\For {$i = 1$ to ID length}
		\State let $v \in S_i$ if the $i^{th} $ bit of $v$'s message $m(v)_i$ is $1$
		\State perform four rounds of \textsc{Decay}$(S_i)$
		\If { $v$ receives a node ID but $m(v)_i = 0$}
			\State $v$ becomes a member of the set $W$ of \textbf{witnesses}
			\State $m(v)\gets $ highest ID $v$ knows
		\EndIf
	\EndFor
	\State $p(v) \gets \textsc{Partial Multi-Broadcast}(W,m(W))$
    \State \textbf{if} $p(v) > \text{ID}(v)$ \textbf{then} $v$ drops out of $C$
	\State \textbf{if} $p(v) = \epsilon$ \textbf{then} $v$ outputs $(m(v),C,1)$, procedure terminates
    \State $v$ outputs $(m(v),C,2)$, procedure terminates
\Else
	\State $v$ outputs $(m(v),C,0)$, procedure terminates
\EndIf
\end{algorithmic}
\end{algorithm}

\begin{theorem}\label{th:VCWOCD}
Algorithm \ref{alg:VCWOCD} performs both \textsc{Verify} and \textsc{Compete} within time $O(T_{BC}) = O(D \log \frac nD+\log^2 n)$, and succeeds with high probability.
\end{theorem}

\begin{proof}
If $C$ is empty, the initial call of \textsc{Partial Multi-Broadcast} will involve no transmissions, so all nodes will output $(\epsilon,\emptyset,0)$, satisfying the requirements.

If $|C| = 1$ then all nodes will receive the ID of its sole member. Therefore there will be no iteration of the for loop in which one node performs \textsc{Decay} and another does not, and so no nodes will become witnesses. Then, the second invocation of \textsc{Partial Multi-Broadcast} will involve no transmissions, so nodes output $(m(v),C,1)$ with $m(v)$ being the ID of $C$'s member, as required.

If $|C|>1$ then there will be at least one pair of neighboring nodes who received different IDs during the first step. Since IDs are $\Theta(\log n)$-bit strings chosen uniformly at random, with high probability every pair of IDs differs on $\Omega(\log n)$ positions (this can easily be seen by taking a union bound over all pairs). For each such position, after four rounds of \textsc{Decay} the node whose received ID has a \textbf 1 in the position will inform the neighboring node with constant probability, by Lemma \ref{corollary:Decay}. Since this event is independent for each of the $\Omega(\log n)$ positions, with high probability the \textsc{Decay} phase succeeds in at least one of them. This results in a node becoming a witness. So, in the second call of \textsc{Partial Multi-Broadcast}, all nodes receive some message and output $(m(v),C,2)$. Clearly this satisfies the conditions of \textsc{Verify}. For \textsc{Compete} we require that at least one node dropped out of $C$. To see this, consider the node $v$ in $C$ which had the lowest ID before the procedure. In the second \textsc{Partial Multi-Broadcast} call it receives an ID from some witness. The witness compared at least two IDs and picked the highest to broadcast. Therefore the ID it picked must have been higher than $v$'s, so $v$ will drop out of $C$.

The running time is dominated by two calls of \textsc{Partial Multi-Broadcast} taking $O(D \log \frac nD+\log^2 n)$ time, and by $O(\log n)$ rounds of \textsc{Decay}, taking $O(\log^2 n)$ time. Therefore total running is $O(D \log \frac nD+\log^2 n)$.
\end{proof}
\subsection{Radio networks with collision detection}
When collision detection is available, we can speed up the \textsc{Verify} procedure by utilizing \emph{beep waves} for broadcasting. As mentioned, our implementations for this setting actually work for the strictly weaker beep model, in which nodes can only transmit a beep or silence.

\subsubsection{Beep Waves}

To study the algorithms for the beep model, we incorporate \emph{beep waves}, first introduced in \cite{-GH13} and formalized in \cite{-CD15}, which are a means of propagating information through the network one bit at a time via waves of collisions.

For our specific purposes, we require procedure $\textsc{Beep-Wave}(S,f)$ which satisfies the following:

\begin{itemize}
\item $S \subseteq V$ is a (possibly empty) set of source nodes;
\item $f:S \rightarrow \{0,1\}^{\ell}$, where $\ell = O(\log n)$ is some message length parameter, is a function giving each source a bit-string to broadcast.
\item Each node interprets a string $m(v)$;
    \begin{itemize}
    \item If $S = \emptyset$, then $m(v) = \epsilon$ (the empty string) for all $v$;
    \item If $S = \{s\}$, for some $s\in V$, then $\forall v \in V$, $m(v)=f(s)$;
    \item If $|S| >1$, then for all $v$, $m(v) \neq \epsilon$. Furthermore, there exists $w \in V$ and two distinct $u, v \in S$ (we allow $w \in \{u, v\}$) such that $m(w) = f(u) \vee f(v) \vee m$, for some bit-string $m$.
    \end{itemize}
\end{itemize}

The last condition may seem convoluted; the reason for it is that, while we cannot guarantee messages are correctly received as in the single source case, we will at least need some means of telling that there were indeed multiple sources. This will be detailed later, but for now we require that, as well as all nodes receiving some non-empty message, at least one receives the logical \textbf{OR} of two source messages, possibly with some extra \textbf{1}s.

We achieve these conditions with Algorithm~\ref{alg:BW}. The intuition behind the algorithm is that a source uses beeps and silence to transmit the $1$s and $0$s of its message respectively (prefixed by a $1$ so that it is obvious when transmission starts), and other nodes forward this pattern, one adjacency layer per time-step, by simply relaying a beep when they hear one. This process takes $O(D + \ell)$ time, where $\ell$, as introduced above, is the maximum length, in bits, of message that is to be broadcast. Since we are here making the restriction $\ell = O(\log n)$, our running time is $O(D + \log n)$.

Note that this method only works in \emph{undirected networks}, since in directed networks it is impossible to prevent beep waves from interfering with subsequent ones, and consequently broadcasting in directed networks in the beep model is much slower.

\begin{algorithm}[ht]
\caption{\textsc{Beep-Wave}$(S,f)$ at a node $v$}
\label{alg:BW}
\small
\begin{algorithmic}
\State initialize $m(v)_i=0$ $\forall i \in [\ell]$
\If {$v \in S$}
	\State $v$ beeps in time-step $0$
	\For {$i = 1$ to $\ell$}
            \State \textbf{if} the $i^{th}$ bit of $f(v)$ is $1$ \textbf{then} $v$ beeps in time time-step $3i$ and $m(v)_i \gets 1$
            \State \textbf{else} $v$ receives a beep in time time-step $3i$ and $m(v)_i \gets 1$
	\EndFor
\Else
	\State {Let $j$ be the time-step in which $v$ receives its first beep}
	\State {$v$ beeps in time-step $j+1$}
	\For {$i = 1$ to $\ell$}
        	\If {$v$ receives a beep in time-step $j+3i$}
            	\State $v$ beeps in time-step $j+3i+1$ and $m(v)_i \gets 1$
            \Else
            	\State $m(v)_i \gets 0$
            \EndIf
	\EndFor
\EndIf
\end{algorithmic}
\end{algorithm}

We prove that the algorithm has the desired behavior:

\begin{lemma}
\label{lem:BW1}
If \textsc{Beep-Wave}$(S,f)$ is run with $S=\{s\}$, then $\forall v \in V$, $ m(v)=f(s)$
\end{lemma}

\begin{proof}
Partition all nodes into layers depending on their distance from the source $s$, i.e., layer $L_i = \{v\in V : dist(v,s) = i\}$. We first note that a node in layer $i$ beeps for the first time in time-step $i$, since this first beep will propagate through the network one layer per time-step. The algorithm then ensures that such a node will beep only in time-steps equivalent to $i \bmod 3$, and only if a beep was heard in the previous step. Since all neighbors of the node must be in layers $i-1, i,$ and $i+1$, only messages from neighbors in layer $i-1$ can be relayed (as these are the only neighbors whose beeps are in time-steps equivalent to $i - 1 \bmod 3$).
It is then easy to see that layers act in unison, and beep if and only if the previous layer beeped in the previous time-step.
\end{proof}


\begin{lemma}
\label{lem:BW2}
If \textsc{Beep-Wave}$(S,f)$ is run with $|S|>1$ then there exists $w \in V$ and two distinct $u, v \in S$ (we allow $w \in \{u, v\}$) such that $m(w) = f(u) \vee f(v) \vee m$, for some bit-string $m$.
\end{lemma}

\begin{proof}
Let $u,v$ be the closest pair of sources in the graph. Let $w$ be the midpoint on the shortest $u \rightarrow v$ path (if the path is of odd length, pick either midpoint arbitrarily).

If $w$ is a source, then we can assume, without loss of generality, that $w = u$ and $v$ is an adjacent source. Then if $f(w)_i = 1$ or $f(v)_i = 1$, $w$ receives or transmits a beep in time-step $3i$ and sets $m(w)_i = 1$, so we are done.

Otherwise, assume without loss of generality that $dist(w,u) \le dist(w,v) \le dist(w,u)+1$, and denote $j := dist(w,u)-1$. $w$ receives its first beep in time-step $j$. Then, since $u$ is the closest source to every node along the shortest $u \rightarrow w$ path and $v$ is the closest source to every node along the shortest $v \rightarrow w$ path, beeps from $u$ and $v$ will always be relayed along these paths. So, if $f(u)_i = 1$, $w$ receives a beep in time-step $j+3i$, and if $f(v)_i = 1$, $w$ receives a beep in time-step $dist(v,w) - 1+3i = j+3i$ or $j+3i+1$ (unless $w$ beeps itself in time-step $j+3i+1$, i.e., it received a beep in time-step $j+3i$). In either case, $m(w)_i$ is set to $1$.
\end{proof}


\begin{lemma}
\label{lem:BW3}
Algorithm \ref*{alg:BW} correctly achieves the Beep-Wave conditions in $O(D+\ell)$ time-steps.
\end{lemma}

\begin{proof}
Correctness for the cases $|S| = 1$ and $|S| >1$ follow from Lemmas \ref{lem:BW1}--\ref{lem:BW2}, and the case $S = \emptyset$ follows since no node ever beeps. To analyze the running time: clearly all sources will have ceased transmission after $O(\ell)$ time, and since beeps are propagated through the network one layer per time-step, it may be a further $D$ time-steps before a source's last beep is heard by the whole network, yielding $O(D+\ell)$ time.
\end{proof}

\subsubsection{Subprocedure Implementation}
We adapt \textsc{Verify} to make use of \textsc{Beep-Wave}, which in these circumstances is faster than \textsc{Partial Multi-Broadcast}.

\begin{algorithm}[h]
\caption{\textsc{Verify}(C) in radio networks with collision detection}
\label{alg:VCD}
\small
\begin{algorithmic}
\State $m(v) \gets \textsc{Beep-Wave}(C,\text{ID*}(C))$
\If {$m(v) \neq \epsilon$}
	\State \textbf{if} $m(v)$ has more than $10\log n$ \textbf{1}s \textbf{then} $v$ becomes a member of witness set $W$
	\State $p(v) \gets \textsc{Beep-Wave}(W,1)$
    \State \textbf{if} $p(v) = \epsilon$ \textbf{then} $v$ outputs $(m(v),1)$, procedure terminates
    \State \textbf{else} $v$ outputs $(m(v),2)$, procedure terminates
\Else
	\State $v$ outputs $(m(v),0)$, procedure terminates
\EndIf
\end{algorithmic}
\end{algorithm}

The idea of this algorithm is similar to that of Algorithm \ref{alg:VCWOCD}, except that now to identify ID clashes we use a property of \textsc{Beep-Wave} rather than rounds of \textsc{Decay}.

To do so we must make a change to candidate IDs: when we have a clash of IDs, what we will get is the logical \textbf{OR} superimposition of the bit-strings, as per the specification of \textsc{Beep-Wave}. Hence, if all IDs have the same number of \textbf{1}s, we will be able to identify such a clash since there will be more \textbf{1}s than a single ID. So, in our verify algorithm we use ID* to mean original ID with $10\log n$ extra bits appended with the purpose of padding the number of \textbf{1}s to exactly $10\log n$.

\begin{theorem}\label{th:VWCD}
Algorithm \ref{alg:VCD} performs \textsc{Verify} within time $O(D+\log n)$ and succeeds with high probability.
\end{theorem}

\begin{proof}
If $C$ is empty, the initial call of \textsc{Beep-Wave} will involve no transmissions, so all nodes will output $(\epsilon,0)$, satisfying the requirements.

If $|C| = 1$ then all nodes will receive the ID of its sole member, which has exactly $10\log n$ \textbf{1}s. Therefore there will be no witnesses, the second invocation of \textsc{Beep-Wave} will involve no transmissions, so nodes output $(m(v),1)$ with $m(v)$ being the ID of $C$'s member, as required.

If $|C|>1$ then by the properties of \textsc{Beep-Wave} at least one node will receive the logical \textbf{OR} superimposition of two IDs (and possibly some other string). Since the IDs are unique w.h.p. (and this remains true with our modification to ID*), this superimposition must have more than $10\log n$ \textbf{1}s. So, at least one node will become a witness and will broadcast during the second invocation of \textsc{Beep-Wave}, and therefore all nodes will output $(m(v),2)$ as required.

The running time is dominated by two calls of \textsc{Beep-Wave} taking $O(D+\log n)$ time.
\end{proof}

We cannot use similar methods to speed up \textsc{Compete}, since the single bit of information that witnesses broadcast is not sufficient to guarantee that at least one candidate will drop out, and we cannot quickly broadcast longer messages from multiple sources without interference. Furthermore, even with a faster \textsc{Compete} sub-procedure, our fixed running time framework would still yield an algorithm slower than the $O((D + \log n\log\log n) \min\{\log\log n, \log\frac{n}{D}\})$ time algorithm of \cite{-GH13}. Therefore we only use our variable-time framework for networks with collision detection.

\section{Running Times}
We can now plug the running times of our implementations of the \textsc{Verify} and \textsc{Compete} sub-procedures (given by Theorems \ref{th:VCWOCD} and \ref{th:VWCD}) into our framework results (Theorems \ref{th:vartm} and \ref{th:fixtm}) to obtain the following leader election results:

\begin{theorem}\label{th:var}
Leader election can be performed in radio networks (either directed or undirected) without collision detection within $O(D \log\frac{n}{D}+\log^2n)$ time \emph{in expectation}, succeeding with high probability.
\end{theorem}

\begin{proof}
Follows immediately from Theorems \ref{th:vartm} and \ref{th:VCWOCD}.
\end{proof}

This running time is asymptotically optimal, and improves by a $\Theta(\log\log n)$ factor over the previous best result of \cite{-GH13}, see Table \ref{table:NOCD}.

\begin{theorem}\label{th:fix}
Leader election can be performed in radio networks (either directed or undirected) without collision detection in $O((D \log\frac{n}{D}+\log^2n)\sqrt{\log n})$ time, succeeding with high probability.
\end{theorem}

\begin{proof}
Follows immediately from Theorems \ref{th:fixtm} and \ref{th:VCWOCD}.
\end{proof}

While this algorithm is slower than that of \cite{-GH13}, it has the benefit of working in directed networks as well as undirected. In directed networks, it is a $\Theta(\sqrt{\log n})$ factor faster than previous results, see Table \ref{table:NOCD}.

Though these results are given by two different frameworks, we can easily combine them into a single algorithm, either by interspersing steps of the two algorithms alternately, or by running Algorithm \ref{alg:LE} for $O((D \log\frac{n}{D}+\log^2n)\sqrt{\log n})$ time and then running Algorithm \ref{alg:LE2}. Consequently, we can achieve good expected and worst case running times concurrently.

For undirected networks in the model with collision detection and in the beep model, we can obtain the following stronger bound (see Table \ref{table:CD} for earlier results).

\begin{theorem}\label{th:cd}
Leader election can be performed in undirected radio networks with collision detection within $O(D + \log n)$ time in expectation, succeeding with high probability.
\end{theorem}

\begin{proof}
Follows immediately from Theorems \ref{th:vartm} and \ref{th:VWCD}.
\end{proof}


\section{Conclusion}
In this paper we presented a framework for leader election in radio networks which yield randomized leader elections taking optimal $O(\text{broadcasting time})$ in expectation, and another which yields algorithms taking fixed $O(\sqrt{\log n})$-times broadcasting time. Both these algorithms succeed with high probability. With this, we provide further evidence that the classical bounds for leader elections using the general randomized framework of simulating single-hop networks with collision detection by multi-hop networks without collision detection due to Bar-Yehuda et al.\ \cite{-BGI91} are suboptimal and can be improved.

There are several important open question in this area. Firstly, can one obtain stronger expected-time bounds in the case of directed graphs with collision detection? Here, the method of beep-waves used as the backbone of Algorithm \ref{alg:VCD} does not provide the same guarantees, as a wave can revisit nodes it already passed through and be interpreted as a new transmission. The best known leader election algorithm in this setting is Algorithm~\ref{alg:LE} with an expected running time of $O(D\log\frac{n}{D} + \log^2 n)$, but the only known lower bound is the same as that of undirected graphs, $\Omega(D+\log n)$, and we would wish to close this gap.

Secondly, what is the worst-case complexity of the leader election problem? The best results in this direction, for undirected graphs, are the algorithms of Ghaffari and Haeupler \cite{-GH13}, with running times of $O(D \log\frac{n}{D} + \log^3 n) \cdot \min\{\log \log n, \log\frac{n}{D}\}$ without collision detection, and $O(D + \log n \log\log n) \cdot \min\{\log\log n, \log\frac{n}{D}\}$ with collision detection, but these running times are off from the respective lower bounds by both an additive term and multiplicative factor. For the directed case our $O((D \log\frac{n}{D} + \log^2 n)\cdot \sqrt{\log n})$-time algorithm is the best known, but is again slower than the $O(D \log\frac{n}{D} + \log^2 n)$ lower bound by a multiplicative factor.

\newcommand{\Proc}{Proceedings of the\xspace}

\newcommand{\STOC}{Annual ACM Symposium on Theory of Computing (STOC)}
\newcommand{\FOCS}{IEEE Symposium on Foundations of Computer Science (FOCS)}
\newcommand{\SODA}{Annual ACM-SIAM Symposium on Discrete Algorithms (SODA)}
\newcommand{\SOCG}{Annual Symposium on Computational Geometry (SoCG)}
\newcommand{\ICALP}{Annual International Colloquium on Automata, Languages and Programming (ICALP)}
\newcommand{\ESA}{Annual European Symposium on Algorithms (ESA)}
\newcommand{\CCC}{Annual IEEE Conference on Computational Complexity (CCC)}
\newcommand{\RANDOM}{International Workshop on Randomization and Approximation Techniques in Computer Science (RANDOM)}
\newcommand{\PODS}{ACM SIGMOD Symposium on Principles of Database Systems (PODS)}
\newcommand{\PODC}{Annual ACM Symposium on Principles of Distributed Computing (PODC)}
\newcommand{\SSDBM}{ International Conference on Scientific and Statistical Database Management (SSDBM)}
\newcommand{\ALENEX}{Workshop on Algorithm Engineering and Experiments (ALENEX)}
\newcommand{\BEATCS}{Bulletin of the European Association for Theoretical Computer Science (BEATCS)}
\newcommand{\CCCG}{Canadian Conference on Computational Geometry (CCCG)}
\newcommand{\CIAC}{Italian Conference on Algorithms and Complexity (CIAC)}
\newcommand{\COCOON}{Annual International Computing Combinatorics Conference (COCOON)}
\newcommand{\COLT}{Annual Conference on Learning Theory (COLT)}
\newcommand{\COMPGEOM}{Annual ACM Symposium on Computational Geometry}
\newcommand{\CPC}{Combinatorics, Probability and Computing}
\newcommand{\DCGEOM}{Discrete \& Computational Geometry}
\newcommand{\DISC}{International Symposium on Distributed Computing (DISC)}
\newcommand{\ECCC}{Electronic Colloquium on Computational Complexity (ECCC)}
\newcommand{\FSTTCS}{Foundations of Software Technology and Theoretical Computer Science (FSTTCS)}
\newcommand{\ICCCN}{IEEE International Conference on Computer Communications and Networks (ICCCN)}
\newcommand{\ICDCS}{International Conference on Distributed Computing Systems (ICDCS)}
\newcommand{\VLDB}{ International Conference on Very Large Data Bases (VLDB)}
\newcommand{\IJCGA}{International Journal of Computational Geometry and Applications}
\newcommand{\INFOCOM}{IEEE INFOCOM}
\newcommand{\IPCO}{International Integer Programming and Combinatorial Optimization Conference (IPCO)}
\newcommand{\ISAAC}{International Symposium on Algorithms and Computation (ISAAC)}
\newcommand{\ISTCS}{Israel Symposium on Theory of Computing and Systems (ISTCS)}
\newcommand{\JACM}{Journal of the ACM}
\newcommand{\JALGORITHMS}{Journal of Algorithms}
\newcommand{\LNCS}{Lecture Notes in Computer Science}
\newcommand{\RSA}{Random Structures and Algorithms}
\newcommand{\SPAA}{Annual ACM Symposium on Parallel Algorithms and Architectures (SPAA)}
\newcommand{\STACS}{Annual Symposium on Theoretical Aspects of Computer Science (STACS)}
\newcommand{\SWAT}{Scandinavian Workshop on Algorithm Theory (SWAT)}
\newcommand{\TALG}{ACM Transactions on Algorithms}
\newcommand{\TCS}{Theoretical Computer Science}
\newcommand{\UAI}{Conference on Uncertainty in Artificial Intelligence (UAI)}
\newcommand{\WADS}{Workshop on Algorithms and Data Structures (WADS)}
\newcommand{\SICOMP}{SIAM Journal on Computing}
\newcommand{\JCSS}{Journal of Computer and System Sciences}
\newcommand{\JASIS}{Journal of the American society for information science}
\newcommand{\PMS}{ Philosophical Magazine Series}
\newcommand{\ML}{Machine Learning}
\newcommand{\DCG}{Discrete and Computational Geometry}
\newcommand{\TODS}{ACM Transactions on Database Systems (TODS)}
\newcommand{\PHREV}{Physical Review E}
\newcommand{\NATS}{National Academy of Sciences}
\newcommand{\MPHy}{Reviews of Modern Physics}
\newcommand{\NRG}{Nature Reviews : Genetics}
\newcommand{\BullAMS}{Bulletin (New Series) of the American Mathematical Society}
\newcommand{\AMSM}{The American Mathematical Monthly}
\newcommand{\JAM}{SIAM Journal on Applied Mathematics}
\newcommand{\JDM}{SIAM Journal of  Discrete Math}
\newcommand{\JASM}{Journal of the American Statistical Association}
\newcommand{\AMS}{Annals of Mathematical Statistics}
\newcommand{\JALG}{Journal of Algorithms}
\newcommand{\TIT}{IEEE Transactions on Information Theory}
\newcommand{\CM}{Contemporary Mathematics}
\newcommand{\JC}{Journal of Complexity}
\newcommand{\TSE}{IEEE Transactions on Software Engineering}
\newcommand{\TNDE}{IEEE Transactions on Knowledge and Data Engineering}
\newcommand{\JIC}{Journal Information and Computation}
\newcommand{\ToC}{Theory of Computing}
\newcommand{\MST}{Mathematical Systems Theory}
\newcommand{\Com}{Combinatorica}
\newcommand{\NC}{Neural Computation}
\newcommand{\TAP}{The Annals of Probability}
\newcommand{\CRYPTO}{Annual International Cryptology Conference (CRYPTO)}
\renewcommand{\CRYPTO}{CRYPTO}
\newcommand{\EUROCRYPT}{Annual International Conference on the Theory and Applications of Cryptographic Techniques}
\renewcommand{\EUROCRYPT}{EUROCRYPT}

\bibliographystyle{plain}
\bibliography{LE}

\end{document}